\documentclass[orivec]{llncs}

\title{Optimal Scheduling for Linear-Rate Multi-Mode Systems} 

\author{Dominik Wojtczak}
\institute{University of Liverpool, Liverpool, UK \\
\email{d.wojtczak@liverpool.ac.uk}}

\usepackage{amsmath}
\usepackage{amssymb}  %
\usepackage[ruled,vlined,linesnumbered,lined]{algorithm2e}
\usepackage{times} %
\usepackage{mathptmx} %

\makeatletter
\renewcommand{\dot}[1]{%
  {\mathop{\kern\z@#1}\limits^{\vbox to-1.4\ex@{\kern-\tw@\ex@
   \hbox{\normalfont\large .}\vss}}}}
\makeatother

\usepackage{latexsym} 
\usepackage{xspace}
\usepackage{stmaryrd}

\newcommand{\tmin}{t_{\text{min}}}
\newcommand{\amax}{a_{\text{max}}}
\newcommand{\andd}{\text{ \scriptsize{\&} }}
\newcommand{\cmin}{c_{\text{min}}}
\newcommand{\cavg}{\mu_\text{avg}}
\newcommand{\cpeak}{\mu_\text{peak}}
\newcommand{\controller}{controller\xspace}
\newcommand{\controllers}{controllers\xspace}
\newcommand{\minsize}{\text{min-size}\xspace}

\newcommand{\FF}{F}
\newcommand{\var}[1]{\bx_{#1}}
\newcommand{\pspace}{{\sc PSpace}\xspace}
\DeclareMathOperator{\sgn}{sgn}

\newcommand{\feasible}{feasible\xspace}
\newcommand{\afeasible}{a feasible\xspace}
\newcommand{\upp}{\point{u}}
\newcommand{\low}{\point{l}}
\newcommand{\good}{implementable\xspace}
\newcommand{\agood}{an implementable\xspace}
\newcommand{\safe}{good\xspace}
\newcommand{\asafe}{a good\xspace}

\DeclareMathOperator{\peakcost}{PeakCost}
\DeclareMathOperator{\avgcost}{AvgCost}
\DeclareMathOperator{\optavgcost}{OptAvgCost}

\newcommand{\norm}[1]{\|#1\|}

\mathchardef\breakingcomma\mathcode`\,
{\catcode`;=\active
  \gdef;{\breakingcomma\discretionary{}{}{}}
}

\newcommand{\point}[1]{{\overline{#1}}}

\DeclareMathOperator*{\diag}{diag}

\newcommand{\px}{\point{x}}
\newcommand{\py}{\point{y}}

\newcommand{\bxdot}{\dot{\mathbf{x}}}
\newcommand{\bx}{\mathbf{x}}
\newcommand{\vbxdot}{\dot{\point{\mathbf{x}}}}
\newcommand{\vbx}{\point{\mathbf{x}}}

\newcommand{\vv}{\vec{v}}

\newcommand{\vf}{\vec{f}}
\newcommand{\vg}{\vec{g}}

\newcommand{\cmax}{{c_{\textrm{max}}}}

\newcommand{\suma}[2]{\Sigma_{#1}^{#2}}

\usepackage{graphicx}
\usepackage{color}
\usepackage{tikz}
\tikzstyle{nloc}=[draw,circle,minimum size=4em,inner sep=0em]
\tikzstyle{mloc}=[draw,circle,minimum size=2em,inner sep=0em]
\tikzstyle{trans}=[-latex, rounded corners]

\usepackage{hyperref}

\newcommand{\set}[1]{\left\{ #1 \right\}}

\newcommand{\seq}[1]{\langle #1 \rangle}
\newcommand{\Rzero}{{\mathbb R}_{\geq 0}}
\newcommand{\Rplus}{{\mathbb R}_{> 0}}
\newcommand{\Nat}{\mathbb N}
\newcommand{\Real}{\mathbb R}
\newcommand{\Int}{\mathbb{Z}}

\newcommand{\Aa}{\mathcal{A}}

\newcommand{\Hh}{\mathcal{H}}

\newcommand{\dg}{$^\circ$C\xspace}

\def\rmdef{\stackrel{\mbox{\rm {\tiny def}}}{=}} %

\begin{document}

\maketitle
\thispagestyle{empty}
\pagestyle{plain}

\allowdisplaybreaks

\begin{abstract}
Linear-Rate Multi-Mode Systems is a model that can be seen both
as a subclass of switched linear systems with imposed global safety constraints
and as hybrid automata with no guards on transitions.
We study the existence and design of a controller for this model that keeps the
state of the system within a given safe set for the whole time.
A sufficient and necessary condition is given for such a controller to exist as well as
an algorithm that finds one in polynomial time. We further generalise the model
by adding costs on modes and present an algorithm that constructs a safe
controller which minimises the peak cost, the average-cost or any cost expressed
as a weighted sum of these two.
Finally, we present numerical simulation results based on our implementation of these
algorithms.
\end{abstract}

 \section{Introduction}
 Optimisation of electricity usage is an increasingly important issue because of
the growing energy prices and environmental concerns. In order to make the whole
system more efficient, not only the average electricity consumption should be
minimised but also its peak demand. The energy produced during the peak times, 
typically occurring in the afternoon due to the heaters or air-conditioning units
being switched on at the same time after people come back from work,
is not only more expensive because the number of consumers outweighs
the suppliers, but also the peaking power plants that 
provide the supply at that time are a lot less efficient. Therefore, the typical
formula that is used for charging companies for electricity is a weighted
average of its peak and average electricity demand \cite{PECO,AS07new}. 
Optimisation of the usage pattern of heating, ventilation and air-conditioning
units (HVAC) not only can save electricity but also contribute to their longer
lifespan, because they do not have to be used just as often.

In~\cite{NBMJ11} Nghiem et al.\ considered a model of an organisation
consisting of a number of decoupled zones whose temperature have to remain within
a specified comfort temperature interval. Each zone has a heater with a number
of possible output settings, but the controller can pick only one of them. That is,
the heater can either be on in that one setting or it has to be off otherwise. 
A further restriction is that only some fixed number of heaters can be on at any
time. The temperature evolution in each zone is governed by a linear
differential equation whose parameters depend on the physical characteristics of the zone, the
outside temperature, the heater's picked setting and whether it is on or
off. The aim is to find a safe controller, i.e. a sequence of time points
at which to switch the heaters on or off, in order for the temperature in each
zone to remain in its comfort interval which is given as the input.
In the end, it was shown that a sufficient condition for such a controller to exist is
whether a simple inequality holds.
This fact can be used to minimise the peak number of heaters used at the same time,
but if heaters can have different costs, then this may not correspond to 
minimising the peak energy cost.  

We strictly generalise the model in~\cite{NBMJ11} and define linear-rate Multi-Mode
Systems (MMS).
The evolution of our system is the same, specifically it consists of a number of
zones, which we will call variables, whose evolution do not directly
influence each other.
However, we do not assume that all HVAC units in the zones are heaters, so the
system can cope with a situation when cooling is required during the day and
heating during the night. Moreover, rather than having all possible combinations
of settings allowed, our systems have a list of allowable joint settings for all
the zones instead; we will call each such joint setting a {\em mode}.
This allows to model specific behaviours, for instance, heat pumps, i.e. when
the heat moves from one zone into another, and central heating that can only
heat all the zones at the same time.
Finally, we will be looking for the actual minimum peak cost without
restricting ourselves to just one setting per heater nor the number of heaters
being switched on at the same time, while keeping the running time polynomial in the number of modes.
We also show how to find the minimum average-cost schedule and finally how to
minimise the energy bill expressed as a weighted sum of the peak and
the average energy consumption.

\vspace{1em}
\noindent {\bf Related work}.
Apart from generalising the model in \cite{NBMJ11}, MMSs can be seen both as
switched linear systems (see, e.g. \cite{GSM01,GS05}) with imposed global
safety constraints or as hybrid automata (\cite{springerlink:10.1007/3-540-57318-6_30,Hen96}) with no guards on
transitions.
The analysis of switched linear systems typically focuses on several forms of
stabilisation, e.g. whether the system can be steered into a given stable region
which the system will never leave again.
However, all these analyses are done in the limit and do not impose any
constraints on the state of the system before it reaches the safe region.
Such an analysis may suffice for systems where the constraints are soft, e.g.
nothing serious will happen if the  temperature in a room will briefly be too
high or too low.
However, it may not be enough when studying safety-critical systems, e.g. when
cooling nuclear rectors. Each zone in an MMS is given a safe value interval
in which the zone has to be at all times.
This causes an interesting behaviour, because even if the system stabilises
while staying forever in any single mode, 
these stable points may be all unsafe and therefore the
controller has to constantly switch between different modes to keep the MMS
within the safety set. For instance, a heater in a room has to constantly switch
itself on and off as otherwise the temperature will become either too high or
too low.
On the other hand, even the basic questions are undecidable for hybrid automata
(see, e.g. \cite{HKPV98}) and therefore MMSs constitute its natural subclass 
with decidable and even tractable safety analysis.

In \cite{ATW12b} we recently studied a different incomparable class of
constant-rate Multi-Mode Systems where in each mode the state of a zone
changes with a constant-rate as opposed to being govern by a linear differential
equation as in linear-rate MMS. Specifically, in ever mode $m \in M$
the value of each variable $\bx_i$ after time $t$ increases by $c^m_i \cdot t$ where \mbox{$c^m_i \in \Real$} is the constant rate of change of $\bx_i$ in mode $m$.
That model was a special case of linear hybrid automata (\cite{springerlink:10.1007/3-540-57318-6_30}), which has
constant-rate dynamics and linear functions as guards on transitions. We showed a polynomial-time algorithm
for both safe controllability and safe reachability questions, as well as finding
optimal safe controllers in the generalised model where each mode has an associated cost per time unit.

There are many other approaches to reduce energy consumption and peak usage in buildings.
One particularly popular one is model predictive control \cite{camacho1999model} (MPC). 
In \cite{OUPAM10} stochastic MPC was used to minimize
building's energy consumption, while in \cite{YBHCBH10} the peak electricity demand reduction was considered.
The drawback of using MPC in our setting is its high computational complexity and the fact it cannot guarantee optimality.   

\vspace{1em}
\noindent {\bf Results}.
The key contribution of the paper is an algorithm for constructing a safe
controller for MMSs with any starting point in the interior of the safety set
that we present in Section \ref{sec:safe}.
Unlike in \cite{NBMJ11}, we not only show a sufficient, but also necessary, condition for such a safe
controller to exist. The condition is a system of linear inequalities that can
be solved using polynomial-time algorithms for linear programming (see, e.g.
\cite{Schr86}) and because that system does not depend on the starting state, we
show that either all points in the interior of the safe set have a safe
controller or none of them has one.
Furthermore, we show that if there is a safe controller then there is a periodic
one with the {\em minimum dwell time}, i.e. the smallest amount of time between two mode switches,
being of polynomial-size.
Such a minimum dwell time may be still too low for practical purposes. However, 
we prove that the problem of checking whether there is a safe controller with
the minimum dwell time higher than $1$ (or for any other constant) is \pspace-hard. 
This means that any approximation of the largest minimum dwell time 
among all safe controllers is unlikely
to be tractable.

In Section \ref{sec:optimal} we generalise the MMS model by associating cost per time unit
with each mode and looking for a safe schedule that minimises the
long-time average-cost. Similarly as before, if there is at least one safe
controller, then the optimal cost do not depend on the starting point
and there is always a periodic optimal controller. In order to prove that the
controller that we construct has the minimum average-cost it is crucial that 
the condition found in Section \ref{sec:safe} is both
sufficient and necessary.
Furthermore, in order to check whether there exists a safe controller with a
peak cost at most $p$, it suffices to check safe controllability for the set of
all modes whose cost do not exceed $p$. This allows us to find the
minimum peak cost using binary search on $p$.

We show that all these periodic safe (and optimal) controllers can be
constructed in time polynomial in the number of modes. However, if one considers the set
of modes to be given implicitly as in \cite{NBMJ11} where each zone has a
certain number of settings and all their possible combinations are allowed  
then the number of modes becomes exponential in the size of the input.
We try to cope with the problem by performing a bottom-up
binary search in order to avoid analysing large sets of modes 
and use other techniques to keep the running time
manageable in practice.

Finally, we show how to find the minimum total cost calculated as a weighted
sum of the peak cost and the average cost. 
The challenging part is that peak cost generally increases when
the set of modes is expanded while the 
average-cost decreases. Therefore, the weighted cost may not be monotone in the
size of set of modes and so binary search may not work and finding its minimum may require
checking many possible subsets of the set of modes.
However, we show a technique how one can narrow down this search significantly to
make it practical even for a large sets of modes.

In the end, we conclude and point out some possible future work in Section \ref{sec:conclusion}.

 \section{Linear Multi-Mode Systems}
 \label{sec:model}

Let us start by setting the notation.
We write $\Nat$ for the set of natural
numbers, $\Nat_{>0}$ for the set of
positive integers and $\Int$ for the set of integers. For a set $X$, let $|X|$ denote the number of elements in $X$. 
Also, we write $\Real$, $\Rzero$, and $\Rplus$ for the sets of
all, non-negative and strictly positive real numbers,
respectively. States of our system will be points in the Euclidean space
$\Real^n$ equipped with the standard {\em Euclidean norm} $\norm{\cdot}$.
By $\px, \py$ we denote points in this state space, by $\vf, \vv$ vectors,
while $\px(i)$ and $\vf(i)$ will denote the $i$-th coordinate of point $\px$
and vector $\vf$, respectively.
For $\Diamond\in \{\leq,<,\geq,>\}$, we write $\px\, \Diamond\, \py$ if
$\px(i)\, \Diamond\, \py(i)$ for all $i$.
For a $n$-dimensional vector
$\vv$ by $\diag(\vv)$ we denote a $n \times n$ dimensional matrix whose diagonal
is $\vv$ and the rest of the entries are $0$. We can now formally define our
model.

\begin{definition}
  \label{def:MMS}
  A linear-rate multi-mode system (MMS) is a tuple $\Hh = (M, N, A, B)$
  where $M$ is a finite nonempty set of {\em modes}, $N$ is the number of
  continuous-time variables in the system, and $A : M \to \Rplus^N, B: M \to
  \Real^N$ give for each mode
  the coefficients of the linear differential equation 
  that govern the dynamics of the system. 
\end{definition}

In all further computational complexity considerations, we assume that all real numbers are rational and
represented in the standard way by writing down the numerator and denominator
in binary. Throughout the paper we will write $a^m_i$ and $b^m_i$ as a shorthand 
for $A(m)(i)$ and $B(m)(i)$, respectively. 

A {\em \controller} of an MMS specifies a timed sequence of mode switches.
Formally, a {\em \controller} is defined as a finite or infinite sequences of
{\em timed actions}, where a timed action $(m, t) \in M \times \Rplus$ is a
tuple consisting of a mode and a time delay.
We say that an infinite \controller  $\seq{(m_1, t_1), (m_2, t_2), \ldots}$
is {\em Zeno} if ${\sum_{k=1}^{\infty} t_k < \infty}$ and  
is {\em periodic} if there exists $l \geq 1$ such that for all $k \geq 1$ we
have $(m_{k}, t_{k}) = (m_{(k \bmod l) + 1}, t_{(k \bmod l) + 1})$.
Zeno \controllers require infinitely many mode-switches within a finite amount of time, and
hence, are physically unrealizable.
However, one can argue that a \controller that switches after $t_k =
{1}/{k}$ amount of time during the $k$-th timed action is also infeasible, because it requires 
the switches to occur infinitely frequently in the limit.
Therefore, we will call a \controller{} {\em \feasible} if its minimum dwell time,
i.e. the smallest amount of time between two mode switches, is positive.
We will relax this assumption and allow for the modes that are not used at all by a feasible controller to 
occur in its sequence of time actions with timed delays equal to $0$, but
we still require any \feasible controller to be non-Zeno. 
For a 
 \controller $\sigma = \seq{(m_1, t_1), (m_2, t_2),
  \ldots}$, we write $T_k(\sigma) \rmdef \sum_{i=1}^k t_i$ for the total time
elapsed up to step $k$ of the \controller $\sigma$,
$T^m_k(\sigma) \rmdef \sum_{i \leq k : m_{i} = m} t_i$ for the total time
spent in mode $m$ up to step $k$, and finally $\tmin(\sigma) = \inf_{\{k\,:\ t_k > 0\}} t_k$ defines the minimum dwell time of
$\sigma$.
For any non-Zeno \controller $\sigma$ we have that $\lim_{k \to \infty} T_k(\sigma) =\infty$
and for any \feasible \controller $\sigma$ we also have
$\tmin(\sigma) > 0$. 
Finally, for any $t \geq 0$ let
$\sigma(t)$ denote the mode the controller $\sigma$ directs
the system to be in at the time instance $t$. Formally, we have $\sigma(t) = m_k$
where $k = \min\{i : t \leq T_i (\sigma)\}$.

The state of MMS $\Hh$ initialized at a starting point $\px_0$ under
control $\sigma$ is a $N$-tuple of continuous-time {\em variables} $\vbx(t)
= (\bx_1(t), \ldots, \bx_{N}(t))$ such that $\vbx(0) = \px_0$ and
$\vbxdot(t) = B(\sigma(t)) - \diag(A(\sigma(t)))\vbx(t)$ holds at any time $t\in\Rzero$.
It can be seen that if $\Hh$ is in mode $m$ during the entire time interval
$[t_0, t_0 + t]$ then the following holds $\bx_i(t_0 + t) =
{b^m_i}/{a^m_i} + (\bx_i(t_0) - {b^m_i}/{a^m_i}) e^{-a^m_i t}$. Notice that this
expression is monotonic in $t$ and converges to ${b^m_i}/{a^m_i}$, because based
on the definition of MMS we have $a^m_i > 0$ for all $m$ and $i$. 

Given a set $S \subseteq \Real^N$ of safe states, we say that a controller
$\sigma$ is $S$-safe for MMS $\Hh$ initialised at $\px_0$ if for all $t \geq 0$
we have $\bx(t) \in S$.
We sometimes say safe instead of {$S$-safe} if $S$ is clear from the context.
In this paper we restrict ourselves to safe sets being hyperrectangles, which
can be specified by giving lower and upper bound value for each variable in the system.
This assumption implies that controller $\sigma$ is $S$-safe iff 
$\bx(t) \in S$ for all $t \in \{T_k(\sigma) : k\geq 0\}$, because 
each $\bx_i(t)$ is monotonic when $\Hh$ remains in the same mode and so
if system is $S$-safe at two time points, the system is $S$-safe in between these two time points as well.  
This fact is crucial to the further analysis and allows us to only focus on 
$S$-safety at the mode switching time points of the controller.
Formally, to specify any hyperrectangle $S$, it suffices to give  
two points $\low, \upp \in \Real^N$, which define the region as follows $S = \{\px : \low \leq \px
\leq \upp\}$. 
The fundamental decision problem for MMS that we solve in this paper is the
following.
\begin{definition}[Safe Controllability]
Decide whether there exists \afeasible
$S$-safe \controller for a given MMS $\Hh$, a hyperrectangular safe set
$S$ given by two points $\low$ and $\upp$ and an initial point $\px_0 \in S$.
\end{definition}

The fact that $a^m_i > 0$ for all $m$ and $i$ make the system {\em stable} in any mode, 
i.e. if the system stays in any fixed mode forever, it will converge to an
equilibrium point.
However, none of these equilibrium points may be $S$-safe and as a result
the controller may need to switch between modes in order to
be $S$-safe. 
We present an algorithm to solve the safe controllability problem in Section
\ref{sec:safe} and later, in Section \ref{sec:optimal}, we generalise the model
to MMS with costs associated with modes and
the aim being finding \afeasible $S$-safe controller with the minimum
average-cost, peak cost, or some weighted sum of
these. As the following example shows, safe controllability can depend on the
starting point if it lies on the boundary of the safe set. 
We will not analyse this special case and assume instead
that the starting point belongs to the interior of the safe set.

It should be noted that the definition of MMS allows for an arbitrary switching between modes.
Restricting the possible order the modes can be used in a timed sequence will be the subject of Corollary \ref{cor:cor}.

\begin{example}
Consider an apartment with two rooms and one heater.
The heater can only heat one room at a time.
When it is off, the room temperature converges to the outside
temperature of 12\dg, while if it is constantly on, the temperature of the room
converges to 30\dg. We assume the comfort temperature to be
between 18\dg and 22\dg.
The table below shows the coefficients $b^m_i$ for all modes $m$ and rooms $i$,
while all the $a^m_i$-s are assumed to be equal to $1$.
Intuitively, when heating room 1 and 2 half of the time each, the temperature in each
room should oscillate around (30\dg+12\dg)/2 = 21\dg and never leave
the comfort zone assuming the switching occurs frequently enough.
We will prove this intuition formally in Section \ref{sec:safe}. 
Therefore, as long as the temperature in one of the rooms is above
18\dg at the very beginning, a safe controller exists. However, if the
temperature in both rooms start at 18\dg (a state which is safe), a safe
controller does not exists, because in every mode the temperature has to drop
in at least one of the rooms and so the state becomes unsafe under any control.
\vskip0.1in
  \begin{center}
    \begin{tabular}{|c|c|c|c|}
      \hline
       Modes  & $m_1$ & $m_2$ & $m_3$ \\
       \hline
      $b^m_1$ (Room 1) & 12  & 30 & 12 \\
      \hline
      $b^m_2$ (Room 2) & 12 & 12 & 30 \\
      \hline
    \end{tabular}
  \end{center}

\end{example}
\vskip0.1in

 \section{Safe Schedulability} 
 \label{sec:safe}
Let us fix in this section a linear-rate MMS $\Hh = (M,N,A,B)$ and a safe set
$S$ given by two points $\low, \upp \in \Real^N$, such that $\low < \upp$ and $S
= \{\px:\low\leq \px \leq \upp\}$. We call any vector $\vf \in \Rzero^M$ such that $\sum_{m\in M}
\vf(m) = 1$ a {\em frequency vector}. 
Also, let us define $\FF_i(\vf, y) := \sum_{m
\in M} \vf(m) (b^m_i - a^m_i y)$.
Notice that for a fixed $i$ and frequency vector $\vf$, function
$\FF_i(\vf,y)$ is continuous and strictly decreasing in $y$. Moreover, 
$\FF_i(\alpha \vf + \beta \vg,y) = \alpha\FF_i(\vf,y) + \beta\FF_i(\vg,y)$.  For
a frequency vector $\vf$, variable $\bx_i$ is called {\em critical} if $\FF_i(\vf, \low_i) = 0$
or $\FF_i(\vf, \upp_i) = 0$ holds.

\begin{definition}
A frequency vector $\vf$ is {\em \safe} if
for every variable $\bx_i$ the following conditions hold (I) $\FF_i(\vf, \low)
\geq 0$, and (II) $\FF_i(\vf, \upp) \leq 0$.
A frequency vector $\vf$ is {\em \good} if
it is \safe and for every variable $\bx_i$ we additionally have
(III) if $\FF_i(\vf, \low) = 0$ then
$\vf(m) = 0$ for every $m \in M$ such that
${b^m_i}/{a^m_i} \neq \low_i$, and (IV) if $\FF_i(\vf,
\upp) = 0$ then 
$\vf(m) = 0$ for every $m \in M$ such that
${b^m_i}/{a^m_i} \neq \upp_i$. 
\end{definition}

\begin{theorem}
\label{thm:safe-to-freq}
If there exists \afeasible $S$-safe controller then there exists \agood frequency
vector.
\end{theorem}
\begin{proof}
  Denote the \feasible $S$-safe controller by $\sigma$.
  Let $f^{(m)}_k = T^m_k(\sigma) / T_k(\sigma)$ be the fraction of the time
  spent by $\sigma$ in mode $m$ up to its $k$-th timed action; 
  note that ${f^{(m)}_k \in [0,1]}$, and 
  ${\sum_{m\in M} f^{(m)}_k = 1}$ for all $k$.
  Let us look at the sequence of vectors
  $\seq{\vf_k \in [0,1]^M}_{k=1}^{\infty}$ where we set $\vf_k(m) =
  f^{(m)}_k$.
  Since this sequence is bounded, by the Bolzano-Weierstrass theorem, there
  exists an increasing integer sequence $j_1, j_2, \ldots$ such that
  $\lim_{k \to \infty} \vf_{j_k}$ exists and let us denote this limit by
  $\vf$. We prove by contradiction that $\vf$ is \agood frequency vector.
  
  First, $\vf$ is a frequency vector as a limit of a
  sequence of frequency vectors. So if it was not \agood one then for some
  variable $\bx_i$ at least one of the following would hold 
  (I) $\FF_i(\vf, \low_i) < 0$, or (II) $\FF_i(\vf, \upp_i) > 0$, or (III) $\FF_i(\vf, \low_i) = 0$ and the set
  $M' := \{m \in M : {\vf(m)>0} \andd {b^m_i}/{a^m_i} \neq
  \low_i\}$ is nonempty, or (IV) $\FF_i(\vf, \upp_i) = 0$ and the set
  $\{m \in M : \vf(m) >0 \andd {b^m_i}/{a^m_i} \neq \upp_i\}$ is
  nonempty. We will consider only cases (I) and (III) as the other two are
  symmetric and their proofs are essentially the same.
  
  Let us first look at case (I). Denote $c:= \FF_i(\vf, \low_i) < 0$. Let
  $\chi_m(t)$ be equal to $1$ if $\sigma(t) = m$ and let it be $0$ otherwise. 
  Notice that $\bx_i(T_k)$,
  the value of the variable $\bx_i$ after the $k$-th timed action of $\sigma$, is equal to 
  $\bx_i(0) + \int_0^{T_k} \bxdot_i(t) dt = \bx_i(0) + \sum_{m
  \in M} \int_0^{T_k} (b^m_i - a^m_i \bx_i(t)) \chi_m(t) dt \leq \upp_i +
  \sum_{m \in M} T^m_k (b^m_i - a^m_i \low_i)$, because if the system is $S$-safe, then for every mode $m$ we have
  $b^m_i - a^m_i \var{i} \leq b^m_i - a^m_i \low_i$. 
  From the definition of $\vf$, for any $\varepsilon > 0$ we can pick
  $K$ such that for all $k > K$ and $m \in M$ we have $|\vf_{j_k}(m) -
  \vf(m)| < \varepsilon$.
  So 
  \begin{align*}
  &\var{i}(T_{j_k}) \leq
  \upp_i + \sum_{m \in M} T^m_{j_k} (b^m_i -
  a^m_i \low_i) 
  = 
  \upp_i + T_{j_k} \sum_{m \in M} \vf_{j_k}(m) (b^m_i -
  a^m_i \low_i)  \\
  &= 
  \upp_i + T_{j_k} \sum_{m\in M} \big(\vf(m) + (\vf_{j_k}(m) -
    \vf(m))\big)
  (b^m_i -
  a^m_i \low_i) \\
  &\leq 
  \upp_i + T_{j_k} (\sum_{m\in M} \vf(m) (b^m_i -
  a^m_i \low_i) + \varepsilon |M| \cmax))  
  =
  \upp_i + T_{j_k} (c + \varepsilon |M| \cmax) 
  \end{align*}
  where $\cmax := \max_{m \in M} |b^m_i - a^m_i \low_i| \geq |\FF_i(\vf,\low_i)|
  > 0$.
  
  If we now set $\varepsilon$ to be $-c/(2|M|\cmax)$, which is $>0$, then
  $\var{i}(T_{j_k}) \leq \upp_i + \frac{1}{2} T_{j_k} c$, and so $\lim_{k
  \to \infty} \var{i}(T_{j_k}) = -\infty$, because $\sigma$ is non-Zeno and $c <
  0$.
  This is a contradiction with the assumption that $\sigma$ is
  $S$-safe, i.e. $\bx_i(t) \geq \low_i$ for all $t\geq 0$.
  
  Now let us move on to case (III). Let $\amax := \max_{m\in M'} {|a^m_i|}$,
  $\cmin := \min_{m\in M'}{|b^m_i - a^m_i \low_i|}$ and $\tmin :=
  \tmin(\sigma)$.
  Of course $\cmin > 0$, because ${b^m_i}/{a^m_i} \neq \low_i$ for $m \in
  M'$ and $\tmin > 0$, because $\sigma$ is \feasible. Let
  $\gamma := \frac{\tmin\cmin}{2 + \tmin\amax}$, which is $>0$. 
\begin{lemma}
\label{lem:half-time}
For at least half of the time duration of every timed action of $\sigma$ which uses mode~$\in M'$,  $\bx_i(t) \geq \low_i + \gamma$
 holds.
\end{lemma}
The proof of Lemma \ref{lem:half-time} can be found in the appendix. We can now proceed similarly as in case (I).
  We have that $\bx_i(T_k)$ is equal to 
  $\bx_i(0) + \int_0^{T_k} \bxdot_i(t) dt = \bx_i(0) +
  \sum_{m \in M} \int_0^{T_k} (b^m_i - a^m_i \bx_i(t)) \chi_m(t) dt \leq \upp_i +
  \sum_{m \in M} T^m_k (b^m_i - a^m_i \low_i) + \sum_{m \in M'}
  -\frac{1}{2}T^m_k a^m_i \gamma $, because using Lemma~\ref{lem:half-time} we
  know that for at least half of the time spent in any mode $m \in M'$ we have $b^m_i -
  a^m_i \bx_i(t) \leq b^m_i - a^m_i (\low_i + \gamma)$ and for the other half and
  any other $m \in M\setminus M'$ we have $b^m_i - a^m_i \bx_i(t) \leq b^m_i - a^m_i \low_i$.
  Again, from the definition of $\vf$, for any $\varepsilon > 0$ we can pick
  $K$ such that for all $k > K$ and $m \in M$ we have $|\vf_{j_k}(m) -
  \vf(m)| < \varepsilon$ and so
  \begin{align*}
  &\var{i}(T_{j_k}) \leq
  \upp_i + \sum_{m \in M} T^m_{j_k} (b^m_i -
  a^m_i \low_i) + \sum_{m \in M'} -\frac{1}{2}T^m_{j_k} a^m_i \gamma \\
  &= 
  \upp_i + T_{j_k} \sum_{m \in M} \vf_{j_k}(m) (b^m_i -
  a^m_i \low_i) - \frac{1}{2} T_{j_k} \gamma \sum_{m \in M'} \vf_{j_k}(m) a^m_i 
  \\
  &= 
  \upp_i + T_{j_k} \sum_{m\in M} \big(\vf(m) + (\vf_{j_k}(m) -
    \vf(m))\big)
  (b^m_i -
  a^m_i \low_i)
  - \frac{1}{2} T_{j_k} \gamma \sum_{m \in M'} \big(\vf(m) + (\vf_{j_k}(m) -
    \vf(m))\big) a^m_i
  \\
  &\leq 
  \upp_i + T_{j_k} (\sum_{m\in M} \vf(m) (b^m_i -
  a^m_i \low_i) + \varepsilon |M| \cmax)) 
   - \frac{1}{2} T_{j_k} \gamma \sum_{m \in M'} \big(\vf(m) - \varepsilon)
  \amax
  \\ 
  &\leq
  \upp_i + T_{j_k} (0 + \varepsilon |M| \cmax + \frac{1}{2} \varepsilon \gamma 
  |M'| \amax - \frac{1}{2} \gamma \sum_{m \in M'} \vf(m)).
  \end{align*}
  If we now set $\varepsilon := \frac{1}{2}\gamma \sum_{m \in M'} \vf(m)
  /(2|M| \cmax + \gamma |M'| \amax)$, which is $ > 0$, then we get $\var{i}(T_{j_k}) \leq \upp_i
  - \frac{1}{4} T_{j_k} \gamma \sum_{m \in M'} \vf(m)$, and so $\lim_{k \to
  \infty} \var{i}(T_{j_k}) = -\infty$, because $\gamma
  \sum_{m \in M'} \vf(m) > 0$ and $\sigma$ is non-Zeno.
  This is a contradiction with the assumption that $\sigma$ is
  $S$-safe.

  Similarly we can show that neither case (II) nor case (IV) can hold which
  finishes the proof that $\vf$ is \agood frequency vector.
\qed
\end{proof}  
  
\begin{theorem}
\label{thm:freq-to-safe}
If there exists \agood frequency
vector then there exists a periodic $S$-safe controller 
for any initial state in the interior of the safety set. 
\end{theorem}
\begin{proof}
Let $\vf$ be the \good frequency vector.  We first remove
from $M$ all modes $m$ such that $\vf(m) = 0$.
We claim that the following periodic controller $\sigma = \seq{(m_k,
t_k)}_{k=1}^\infty$ with period $|M|$ is $S$-safe for sufficiently small $s$: 
$m_k$ = $(k \bmod |M|) + 1$ and $t_k = \vf(m_k) \cdot s$. 
As we already know it suffices to check $S$-safety of the system at time points
$T_k$ for all $k$. We will focus here on checking just the lower bound,
$\vbx(T_k) \geq \low$, because the estimations concerning the upper bound are
very similar.
Note that for any variable $\bx_i$ we have
  \begin{align*}  
  &\bx_i(T_1) = \frac{b^{m_1}_i}{a^{m_1}_i} + \left(\px_0(i) -
  \frac{b^{m_1}_i}{a^{m_1}_i}\right) e^{-a^{m_1}_i t_1}, \text{ and } \\ 
  \bx_i(T_2) = 
  &\frac{b^{m_2}_i}{a^{m_2}_i} +
  \left(\frac{b^{m_1}_i}{a^{m_1}_i} - \frac{b^{m_2}_i}{a^{m_2}_i}\right) e^{-a^{m_2}_i
  t_2} + \left(\px_0(i) - \frac{b^{m_1}_i}{a^{m_1}_i}\right) e^{-a^{m_1}_i t_1 - a^{m_2}_i
  t_2}
  \end{align*}
  and further by induction we get
  \begin{align}  
  \label{eq:k-step}
  \bx_i(T_k) = \frac{b^{m_k}_i}{a^{m_k}_i}& +
  \sum_{n = 1}^{k-1}
  \left(\frac{b^{m_n}_i}{a^{m_n}_i} - \frac{b^{m_{n+1}}_i}{a^{m_{n+1}}_i}\right)
  e^{-\sum_{j=k-n+1}^{k} a^{m_j}_i t_j}
  + \left(\px_0(i) -
  \frac{b^{m_1}_i}{a^{m_1}_i}\right) e^{-\sum_{j=1}^{k} a^{m_j}_i t_j}
  \end{align}  
Now, because $\vf$ is \good, if $\FF_i(\vf, \low_i) = 0$ then it has to be
${b^m_i}/{a^m_i} = \low_i$ for all $m$. In such a case, it is easy to see 
from equation (\ref{eq:k-step}) that $\bx_i(T_k) > \low_i$ for all $k$, because
$\px_0(i) > \low_i$. Therefore, we can assume $\FF_i(\vf, \low_i) > 0$.
Let $\px_l \rmdef \vbx(T_{l|M|})$ for all $l \in \Nat$, $\suma{k}{l} \rmdef
\sum_{j=k}^{l} a^{m_j}_i t_{j} = s \cdot \sum_{j=k}^{l} a^{m_j}_i \vf(m_j)$, $\alpha(s) \rmdef e^{-\suma{1}{|M|}}$, and 
$$\beta(s) \rmdef 
\frac{b^{m_{|M|}}_i}{a^{m_{|M|}}_i} +
  \sum_{n = 1}^{|M|-1}
  \left(\frac{b^{m_n}_i}{a^{m_n}_i} - \frac{b^{m_{n+1}}_i}{a^{m_{n+1}}_i}\right)
  e^{-\suma{|M|-n+1}{|M|}} -
  \frac{b^{m_1}_i}{a^{m_1}_i} e^{-\suma{1}{|M|}}.$$
Notice that since $\sigma$ is periodic with period $|M|$ from
equation (\ref{eq:k-step}) we can deduce 
$\px_{l+1}(i) = \alpha(s) \px_l(i) + \beta(s)$ for all $l$. Now, because
$0 < \alpha(s) < 1$ for all $s$, sequence $\px_l(i)$ converges monotonically
to $\beta(s) / (1-\alpha(s))$ as $l \to \infty$ for any
initial value $\px_0(i)$.

We will now find $s$ of polynomial size such that  
$\beta(s) / (1-\alpha(s)) \geq \low_i$. The last condition is equivalent to
$\beta(s) - \low_i + \alpha(s) \low_i \geq 0$, because $1-\alpha(s) > 0$. 
It is well-known that $1- x \leq e^{-x} \leq 1 - x + x^2$ for all $x \geq 0$.
Let us also denote $d := |\low_i| + |\upp_i| + 2 \cdot \max_{m \in M} |b^m_i /
a^m_i|$.
Notice that for all $m, m' \in M$ we have  
$|b^m_i/a^m_i - b^{m'}_i/a^{m'}_i| \leq d$,
$|\low_i - b^{m}_i/a^{m}_i| \leq d$, as well as $|\px_0(i) - b^{m}_i/a^{m}_i|
\leq d$. Therefore, $  \beta(s) - \low_i + \alpha(s) \low_i =  
(\frac{b^{m_{|M|}}_i}{a^{m_{|M|}}_i} - \low_i) +
  \sum_{n = 1}^{|M|-1}
  (\frac{b^{m_n}_i}{a^{m_n}_i} - \frac{b^{m_{n+1}}_i}{a^{m_{n+1}}_i})
  e^{-\suma{|M|-n+1}{|M|}} + 
  (\low_i - \frac{b^{m_1}_i}{a^{m_1}_i}) e^{-\suma{1}{|M|}}
  \geq 
  (\frac{b^{m_{|M|}}_i}{a^{m_{|M|}}_i} - \low_i) +
  \Big(\sum_{n = 1}^{|M|-1}
  (\frac{b^{m_n}_i}{a^{m_n}_i} - \frac{b^{m_{n+1}}_i}{a^{m_{n+1}}_i})
  (1-\suma{|M|-n+1}{|M|}) - d \cdot (\suma{|M|-n+1}{|M|})^2 \Big) + 
  (\low_i - \frac{b^{m_1}_i}{a^{m_1}_i}) (1-\suma{1}{|M|}) -
  d \cdot (\suma{1}{|M|})^2 
  =
   \sum_{j=1}^{|M|} b^{m_j}_i t_{m_j} - \low_i \suma{i}{|M|} - \sum_{n =
   1}^{|M|-1} d \cdot (\suma{|M|-n+1}{|M|})^2 - d \cdot (\suma{1}{|M|})^2 
  \geq
   s \left(\sum_{j=1}^{|M|} \vf(m_j) b^{m_j}_i - \sum_{j=1}^{|M|} \vf(m_j)
   a^{m_j}_i \low_i\right)  - d \cdot |M|(\suma{1}{|M|})^2
  =
  s \FF_i(\vf, \low_i) - s^2 d |M|(\sum_{j=1}^{|M|} \vf(m_j) a^{m_j}_i)^2
  $. If we now set $s := \frac{\FF_i(\vf, \low_i)}{d |M|(\sum_{m}
  \vf(m) a^{m}_i)^2}$ then the last expression will be $\geq 0$. 
  Notice that this
  bound does not depend on the order of the modes in the period nor on the
  starting state. So if for
  any $k < |M|$ we repeat this estimation for the initial point $\vbx(T_k)$, 
  controller $\sigma'(t) := \sigma (t + T_k)$ and 
  exactly the same $s$, the value of $\bx_i(T_{l|M|+k} - T_k)$ 
  under control $\sigma'$ will also monotonically converge to some value $\geq
  \low_i$ as $l \to \infty$.
  Therefore, as long as $\vbx(T_k)$ is a $S$-safe for all $k < |M|$ for the just selected $s$, 
  all states of the system that follow will be $S$-safe as well.
  Now, if we repeat the same analysis for the upper bound $\upp_i$ then we
  would get an expression $s := \frac{-\FF_i(\vf, \upp_i)}{d |M|(\sum_{m}
  \vf(m) a^{m}_i)^2}$, so it suffices to set $s$ to be the minimum of these two.

  Now, to find $s$ such that the system is $S$-safe for the
  first $|M|$ steps, we can
  estimate 
  $\bx_i(T_k)$ to be $\geq \px_0(i) - T_k \max_m |b^m_i -
  a^m_i \px(0)|$ and $\leq \px_0(i) + T_k \max_m |b^m_i -
  a^m_i \px(0)|$. We have $T_k \leq s$ for
  $k < |M|$ from the definition of $\sigma$ and so it suffices to set $s:= \min
  \{\upp_i - \px_0(i), \px_0(i) - \low_i\} / \max_m |b^m_i - a^m_i \px_0(i)|$
  if $\max_m |b^m_i -
  a^m_i \px(0)| \neq 0$ and otherwise set $s$ to an arbitrary high value  
  in order for the variable $\bx_i$ to be $S$-safe in the first $|M|$ steps. 
  
  Finally, if we pick the minimum value from these estimates on $s$ over all
  possible variables $\bx_i$, we will guarantee that the system is both $S$-safe in the first
  $|M|$ steps as well as after that, because $\bx_i(T_{l|M|+k})$ will
  monotonically converge for every fixed $i$ and $k < |M|$ to a safe state as $l \to
  \infty$. 
\qed
\end{proof}

\begin{algorithm}[p]
\label{alg:safe-schedule}
\caption{Finds a $S$-safe \feasible controller from a given $\px_0 \in S$.} 
\KwIn{MMS $\Hh$, two points $\low$ and $\upp$ that define a 
hyperrectangle $S = \{\px : \low \leq \px \leq \upp\}$ and an initial point
$\px_0 \in S$ such that $\low < \px_0 < \upp$.} 
\KwOut{NO if no $S$-safe
\feasible controller exists from $\px_0$, and a periodic such controller, otherwise.
} 

$I:=M$;

Check whether the following linear program is satisfiable for some frequency
vector $\vf$:
\nllabel{alg-line:safe-cond}    
\begin{align*}
  		\FF_i(\vf, \low_i) &\geq {0} \text{ for
  		all $i\in I$}\nonumber\\
  		\FF_i(\vf, \upp_i) &\leq {0} \text{ for
  		all $i\in I$}. \label{eqn1}
\end{align*}
\nllabel{alg-line:safe-safe}
\If{no satisfying assignment exists}{\Return NO}
Let $\vf^*$ be any frequency vector of polynomial size that satisfies conditions in step \ref{alg-line:safe-cond}.

\Repeat{no mode was removed from $M$ in this iteration}{
\nllabel{alg-line:safe-loop1}    
  		\ForEach{$j \in I$}{
  			Check whether the following linear program is satisfiable for some
  		frequency vector $\vf$:
\nllabel{alg-line:safe-good-freq}
			\begin{align*}
  				\FF_i(\vf, \low_i) &\geq {0} \text{ for
  				all $i\in I\setminus\{j\}$}\nonumber\\
  				\FF_i(\vf, \upp_i) &\leq {0} \text{ for
  				all $i\in I\setminus\{j\}$} \\
  				\FF_j(\vf, \low_j) &> {0} \nonumber \text{ and } \FF_j(\vf, \upp_j) < {0}
  				\nonumber.
			\end{align*}
			\If{no satisfying assignment exists}{
				If $\FF_j(\vf^*, \low_i) = 0$, remove all modes for which $b^m_i/a^m_i \neq
				\low_i$ and otherwise remove all modes for which $b^m_i/a^m_i \neq \upp_i$.

				Remove $j$ from $I$.
			}
		}
}
\nllabel{alg-line:safe-loop2}
Check whether the following linear program is
satisfiable for any frequency vector $\vf$:
\begin{align*}
  		\FF_i(\vf, \low_i) &> {0} \text{ for
  		all $i\in I$}\\
  		\FF_i(\vf, \upp_i) &< {0} \text{ for
  		all $i\in I$}.
\end{align*}
\nllabel{alg-line:safe-LP2}
\If{no satisfying assignment exists or $M = \emptyset$}{\Return NO}
Let $\vf_*$ be any frequency 
vector of polynomial size that satisfies conditions in step \ref{alg-line:safe-LP2}.

Let 
$$
s:=\min_{i\in I} \left(
\frac{\min\{x_0(i) - \low_i, \upp_i - x_0(i)\}}{\max_m |b^m_i -
a^m_i x_0(i)|},
\frac{\min(\FF_i(\vf_*,\low_i),\FF_i(\vf_*,\upp_i))
}{(|\low_i| + |\upp_i| + 2\cdot \max_m |{b^m_i}/{a^m_i}|) (\sum_m a^m_i
\vf_*(m))^2}\right).
$$

{\bf return} the following periodic controller with period $|M|$: $m_k$ = $(k
  \bmod |M|) + 1$ and $t_k = \vf_*(m_k) \cdot s$.
\end{algorithm}

\begin{theorem}
Algorithm \ref{alg:safe-schedule} returns in polynomial time a \hbox{$S$-safe} feasible
controller from $\px_0$ if there exists one.
\end{theorem}
\begin{proof}
We first need the following lemma whose proof can be found in the appendix.
\begin{lemma}
\label{lem:critical}
Either there is a variable which is critical for all \safe frequency vectors or
there is \asafe frequency vector in which no variable is critical.
\end{lemma}

Now, let $\sigma$ be the controller returned by Algorithm
\ref{alg:safe-schedule}.
Notice that the frequency vector $\vf_*$ the controller $\sigma$ is based on is
\good, because $\vf_*$ satisfies the constraints at line \ref{alg-line:safe-LP2} which imply the
conditions (I) and (II) of $\vf$ being \good and from Lemma \ref{lem:critical}
it follows that all modes that could violate the conditions (III) and (IV) were
removed in the loop between lines \ref{alg-line:safe-loop1}--\ref{alg-line:safe-loop2}. Moreover, constant $s$
used in the construction of $\sigma$ is exactly the same as the one used in
Theorem \ref{thm:freq-to-safe} which guarantees $\sigma$ to be $S$-safe.

On the other hand, from Theorem \ref{thm:safe-to-freq}, if there exists a
feasible $S$-safe controller then there also exists \agood frequency vector $\vf$. Such a vector will
satisfy the constraints of being \safe at line \ref{alg-line:safe-safe} of the
algorithm. In the loop between the lines \ref{alg-line:safe-loop1}--\ref{alg-line:safe-loop2},
all variables that are critical in $\vf$ are first checked whether they satisfy conditions (III) and (IV),
and they will satisfy them because $\vf$ is \good, and after that these critical variables are removed.
Finally, $\vf$ consisting of just the remaining variables will satisfy the
constraints at line \ref{alg-line:safe-good-freq} of being \good with no critical variables. 
Therefore, Algorithm \ref{alg:safe-schedule} will always return a controller if
there exists a $S$-safe one.

It is easy to see that Algorithm \ref{alg:safe-schedule} runs in
polynomial time, because at least one critical variable is removed in each iteration of the loop between lines
\ref{alg-line:safe-loop1}--\ref{alg-line:safe-loop2}, one iteration
checks at most $N$ remaining variables, and each such a check requires calling a
linear programming solver which runs in polynomial time. 
Finally, steps 4 and 13 of Algorithm \ref{alg:safe-schedule} are achievable, 
because if a linear program has a solution then it has
a solution of polynomial size (see, e.g. \cite{Schr86}). 
This shows that the size of the returned controller is always polynomial.
\qed 
\end{proof}

Notice that the controller returned by Algorithm \ref{alg:safe-schedule} has a
polynomial-size minimum dwell time. We do not know whether 
finding a safe controller with the largest possible dwell time is decidable,
nor is checking whether such a minimum dwell time can be greater than $\geq 1$. 
We now show that the last problem is \pspace-hard, so it is unlikely to be tractable.
The details of the proof are in the appendix. Finally, this also implies \pspace-hardness of checking whether a safe controller exists in the case
the system is controlled using a digital clock, i.e. when all timed delays have to be a multiple of some given sampling rate $\Delta > 0$. 

\begin{theorem}
\label{thm:min-dwell-time-hardness}
For a given MMS $\Hh$, hyperrectangular safe set $S$ described by two points
$\low, \upp$, starting point $\px_0 \in S$, checking whether
there exits a $S$-safe controller with minimum dwell time $\geq 1$ is \pspace-hard.
\end{theorem}
\begin{proof}{\em (Sketch)}
The proof is similar to the \pspace-hardness proof in \cite{ATW12b} of
the discrete-time reachability in constant-rate MMS that reduces from 
the acceptance problem for linear bounded automata (LBAs), but our reduction 
is a lot more involved, because of the differences in the dynamics of the system.
For instance, we deal with a decision problem for the minimum dwell time of a safe continuous-time controller
instead of a discrete-time one.
Also, unlike for constant-rate MMS, 
there is no possibility to keep the value of any variable constant over time
regardless of its current value. 
To overcome this problem, we will
take advantage of the fact that for every LBA and input word, there exists
an exponential upper bound on the number of steps this LBA can take before the input word is accepted.
\qed
\end{proof}

Notice that the periodic controller returned by Algorithm \ref{alg:safe-schedule} just cycles forever over the set of modes in some fixed order
which can be arbitrary. This allows us to extend the model by specifying an initial mode $m_0$ and a directed graph $G \subseteq M \times M$, 
which specifies for each mode which modes can follow it. Formally, we require any controller 
$\seq{(m_1, t_1), (m_2, t_2), \ldots}$ to satisfy $(m_i,m_{i+1}) \in G$ for all $i \geq 1$ and $m_1 = m_0$.
The proof is in the appendix.

\begin{corollary}
\label{cor:cor}
  Deciding whether there exists \afeasible $S$-safe \controller for a given MMS $\Hh$
  with a mode order specification graph $G$, initial mode $m_0$, a hyperrectangular safe set
  $S$ given by two points $\low$ and $\upp$ and an initial point $\low < \px_0 < \upp$ can be done in 
  polynomial time.  
\end{corollary}
  \section{Optimal Schedulability}
  \label{sec:optimal}
  In this section we extend our results on $S$-safe controllability of MMS to a
model with costs per time unit on modes. We will call this model {\em priced linear-rate multi-mode systems}.
The aim is to find an $S$-safe
controller with the minimum cost where the cost is either defined as the peak
cost, the (long-time) average cost or a weighted sum of these.
\begin{definition}
  \label{def:pMMS}
  A priced linear-rate multi-mode system (MMS) is a tuple $\Hh = (M, N, A, B, \pi)$
  where $(M, N, A, B)$ is a MMS and ${\pi: M \to \Rzero}$ is a cost function
  such that $\pi(m)$ characterises the cost per-time unit of staying in mode
  $m$.
\end{definition}

We define the (long-time) average cost of an infinite controller $\sigma = \seq{(m_1, t_1),
(m_2, t_2), \ldots}$ as the long-time average of the cost per time-unit over time, i.e.
$$\avgcost(\sigma) \rmdef \limsup_{k \to \infty} \frac{\sum_{i=1}^k \pi(m_i) \cdot
t_i} {\sum_{i=1}^k t_i}.$$
For the results to hold it is crucial that $\limsup$ is used in this definition
instead of $\liminf$. In the case of minimising the average cost, it is more natural to
minimise its $\limsup$ anyway, which intuitively is its reoccurring maximum value. 
On the other hand, the peak cost is simply defined as $\peakcost(\sigma) \rmdef \sup_{\{k\,:\ t_k > 0\}}
\pi(m_k).$

We will try to answer the following question for priced MMS.
\begin{definition}[Optimal Controllability]
  Given a priced MMS $\Hh$, a hyperrectangular safe set $S$ defined
  by two points $\low$ and $\upp$, an initial point $\px_0 \in S$ such that
  $\low < \px_0 < \upp$, and constants $\cavg, \cpeak \geq 0$, find an $S$-safe
  controller $\sigma$ with the minimum value of $\cavg \avgcost(\sigma) + \cpeak
  \peakcost(\sigma)$.
\end{definition}

The following example shows that 
such a weighted cost does not always increase with 
the increase in the peak cost.

\begin{example}
The table below shows the values
of $b^m_i$ for each mode $m \in M = \{m_1, m_2, m_3\}$ and variable $\bx_i$ 
as well as the cost of each
mode. We assume that all $a^m_i$-s are equal to $1$. The
safe value interval for each variable is $[0,1]$, i.e. $\low_i = 0$, $\upp_i
= 1$ for all $i$.
\begin{center}
    \begin{tabular}{|c|c|c|c|c|}
      \hline
      Modes & $m_1$ & $m_2$ & $m_3$ & $m_4$\\
       \hline
${b^m_1}$ & $-1$ & $2$ & $-1$ & $5$\\
       \hline
${b^m_2}$ & $-1$ & $-1$ & $2$ & $5$\\
       \hline
$\pi$ (cost) & $0$ & $3$ & $3$ & $4$\\
       \hline
    \end{tabular}
  \end{center}
  One can compute that the optimal average cost of any $S$-safe controller which uses only modes from
  $M' = \{m_1,m_2,m_3\}$ is equal to $2$ and that average cost is achieved when the frequency of each mode from $M'$
  is equal to $\frac{1}{3}$. At the same time, the peak cost of that controller is $3$.  
  On the other hand, there is a $S$-safe controller for the whole set of modes $M$ 
  with peak cost $4$ and average cost just
  $\frac{5}{6}$, when the frequency of mode $m_1$ is $\frac{5}{6}$ and the frequency of mode
  $m_4$ is $\frac{1}{6}$. If we assume that the weighted cost of a controller
  $\sigma$ is $\peakcost(\sigma) + \avgcost(\sigma)$, then clearly the second controller 
  has a lower weighted cost although it has a higher peak cost.
\end{example}

\begin{algorithm}
\label{alg:opt-schedule}
\caption{Finds an optimal $S$-safe \feasible controller from a given $\px_0 \in
S$.} \KwIn{A priced MMS $\Hh$, two points $\low$ and $\upp$ that define a 
hyperrectangle $S = \{\px : \low \leq \px \leq \upp\}$ and an initial point
$\px_0 \in S$ such that $\low < \px_0 < \upp$, and constants $\cavg$ and
$\cpeak$ which define the weighted cost of a controller.} \KwOut{NO if no $S$-safe
\feasible controller exists from $\px_0$, and an periodic such controller $\sigma$
for which $\cpeak \peakcost(\sigma) + \cavg \avgcost(\sigma)$ is
minimal, otherwise.
} 

$\minsize:=1$;

\Repeat{$\minsize < |M|$ {\bf and} the call returned NO}{

$\minsize := 2\cdot\minsize$;

Pick minimal $p$ such that $M_{\leq p}$, the set of all modes with cost at
most $p$, has size at least $\minsize$.

Call Algorithm \ref{alg:safe-schedule} for the set of modes $M_{\leq p}$.
}

\If{the last call to Algorithm \ref{alg:safe-schedule} returned NO}
{{\bf return} NO.}

Perform a binary search to find the minimal $p$ such that $M_{\leq p}$ is
feasible using the just found upper bound on the minimal feasible set of modes.

Modify Algorithm
\ref{alg:safe-schedule} by adding
the objective function
{\em Minimise $\sum_{m\in M}{\vf_m \pi(m)}$}
to the linear program at line 
\ref{alg-line:safe-LP2}. Let $\optavgcost(M')$ 
be the value of this objective when Algorithm \ref{alg:safe-schedule} is
called for the set of modes $M'$. 

$p':=p'+\frac{\cavg}{\cpeak}\optavgcost(M_{\leq p})$;

\Repeat{$p'$ decreases}{

$p':=p'+\frac{\cavg}{\cpeak}\left(\optavgcost(M_{\leq p}) - \optavgcost(M_{\leq
p'})\right)$; }

Pick a peak value $p^* \in [p, p']$ for which 
$\cpeak p^* + \cavg \optavgcost(M_{\leq p^*})$ is
the smallest.

{\bf return} the periodic controller returned by
the modified version of Algorithm \ref{alg:safe-schedule} called for the set of modes $M_{\leq p^*}$.
\end{algorithm}

The algorithm that we define is designed to cope with 
systems where the set of modes is large and given implicitly like in \cite{NBMJ11}, where
the input is a list of heaters with different output levels and energy costs.
Each heater is placed in a different zone and any possible on/off combination of the heaters
gives us a different mode in our setting, which leads to 
exponentially many modes in the size of the input.
The cost of a mode is  the sum
of the energy cost of all heaters switched on in that particular mode.
We try to deal with this setting by using binary search and a specific narrowing down technique to consider only
the peak costs for which the weighted cost can be optimal.
Unfortunately our algorithm will not run in time polynomial in the number of heaters, but the techniques used 
can reduce the running time in practice.
If we assume that modes are given explicitly as the input then there is a much simpler algorithm
which runs in polynomial time and is presented in Appendix \ref{app:alg} as Algorithm \ref{alg:opt-schedule-simpler}. 
Let us now fix a MMS with costs $\Hh = (M, N, A, B, \pi)$, the safe set
$S$ and a starting point $\px_0$ in the interior of $S$.

\begin{theorem}
Algorithm \ref{alg:opt-schedule} finds a $S$-safe feasible periodic controller that optimises the
weighted cost defined by the peak and average cost coefficients $\cpeak$
and $\cavg$.
\end{theorem}
\begin{proof}
Let $M_{\leq
p}$ denote the set of modes with cost at most $p$.
First, to find the minimum peak cost among all $S$-safe controllers we can first
order all the modes
according to their costs and then 
the algorithm makes a binary search on the
possible peak cost $p$, i.e. guesses initial $p$ and checks whether $M_{\leq
p}$ has a $S$-safe controller; if it
does not then it doubles the value of $p$ and if it does then halves the value of $p$. 
In may be best to start with a small value of $p$ first, because the bigger $p$ is, the bigger is the set
of modes and the slower is checking its feasibility.

Second, to find the minimum average cost among all $S$-safe controllers, notice that the average
cost of the periodic controller returned by Algorithm \ref{alg:safe-schedule} based on the
frequency vector $\vf_*$ is $\sum_{m\in M} \vf_*(m) \pi(m)$.
Therefore, if we find \agood frequency vector which minimises that value, then we
will also find a safe controller with the minimum average-cost among
all periodic safe controllers. This can be easily done by adding 
the objective {\em Minimize} $\sum_m \vf_*(m) \pi(m)$ to the linear program at line
\ref{alg-line:safe-LP2} of Algorithm \ref{alg:safe-schedule}.
However, using similar techniques as in Theorem \ref{thm:safe-to-freq} we can show that no other controller can have a lower average-cost.
The key observation is the fact that $\avgcost(\sigma) = \limsup_{k \to \infty}
\sum_{m \in M} f^{(m)}_k \pi(m) \geq \limsup_{k \to \infty} \sum_{m \in M} f^{(m)}_{j_k} \pi(m) = \sum_m
\vf(m) \pi(m)$ where just like in the proof of Theorem \ref{thm:safe-to-freq},
$f^{(m)}_k$ is the frequency of being in mode $m$ up to the $k$-th timed action and
$\seq{j_k}_{k\in\Nat}$ defines a subsequence of $f^{(m)}_k$ that converges for every
$m$. The second inequality holds because a $\limsup$ of a subsequence is at most equal to the $\limsup$ of the whole sequence.

For any set of modes $M' \subseteq M$, 
let $\optavgcost(M')$ denote the minimum average cost when only modes in $M'$ can be used.
Now, if $\cpeak = 0$ then it suffices to compute the optimal
average cost for the whole set of modes to find the minimum weighted cost. Otherwise, 
to find a safe controller with the minimum value of $\cpeak
\peakcost(\sigma) + \cavg \avgcost(\sigma)$ the algorithm first
finds a feasible set of modes with the minimum peak cost and let us denote that peak cost
by $p_{min}$. If $\cavg = 0$ then this suffices. Otherwise, observe that from the definition the
cost of each mode is always nonnegative and so the average cost has to be as well. Even
if we assume that the average cost is equal to $0$ for some larger set of modes with peak cost $p$, the weighted cost
will at least be equal to $\cpeak p$ as compared to $\cpeak p_{min} + \cavg
\optavgcost(M_{\leq p_{min}})$, which gives us an upper bound on the maximum value of $p$
worth considering to be $p' = p_{min} +
\frac{\cavg}{\cpeak}(\optavgcost(M_{\leq p_{min}}))$. But now we can check the
actual value of $\optavgcost(M_{\leq p'})$ instead of assuming it is $=0$ and calculate again a new bound on the maximum peak value
worth considering and so on.
To generate modes on-the-fly in the order of increasing costs, we can use Dijkstra algorithm with a priority queue.
\qed
\end{proof}

   \section{Numerical Simulations}
   \label{sec:experimental}
   \begin{figure}[htbp]
\hskip-0.8in
   \includegraphics[width=0.65\hsize, angle=0]{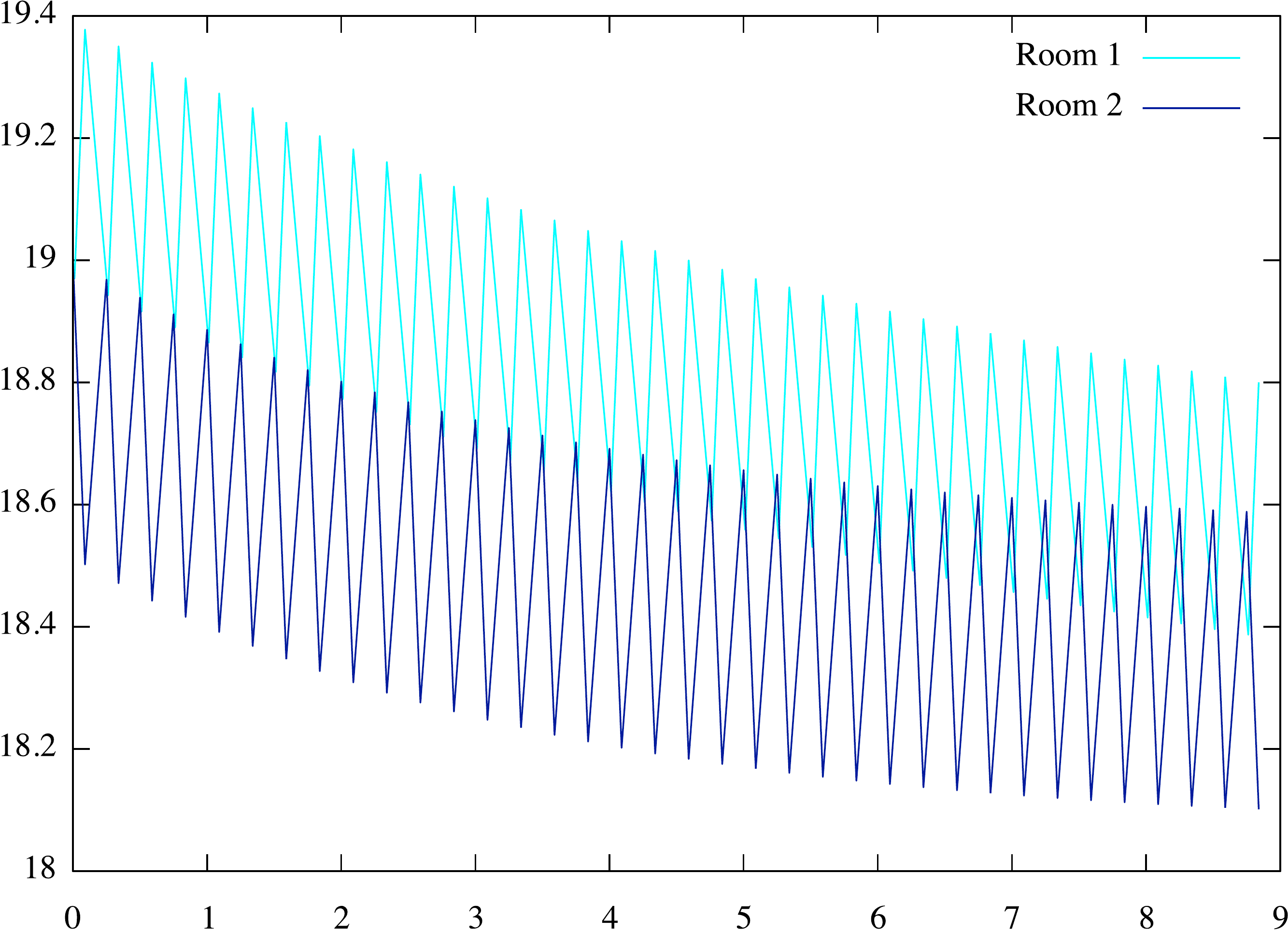}
   \includegraphics[width=0.65\hsize, angle=0]{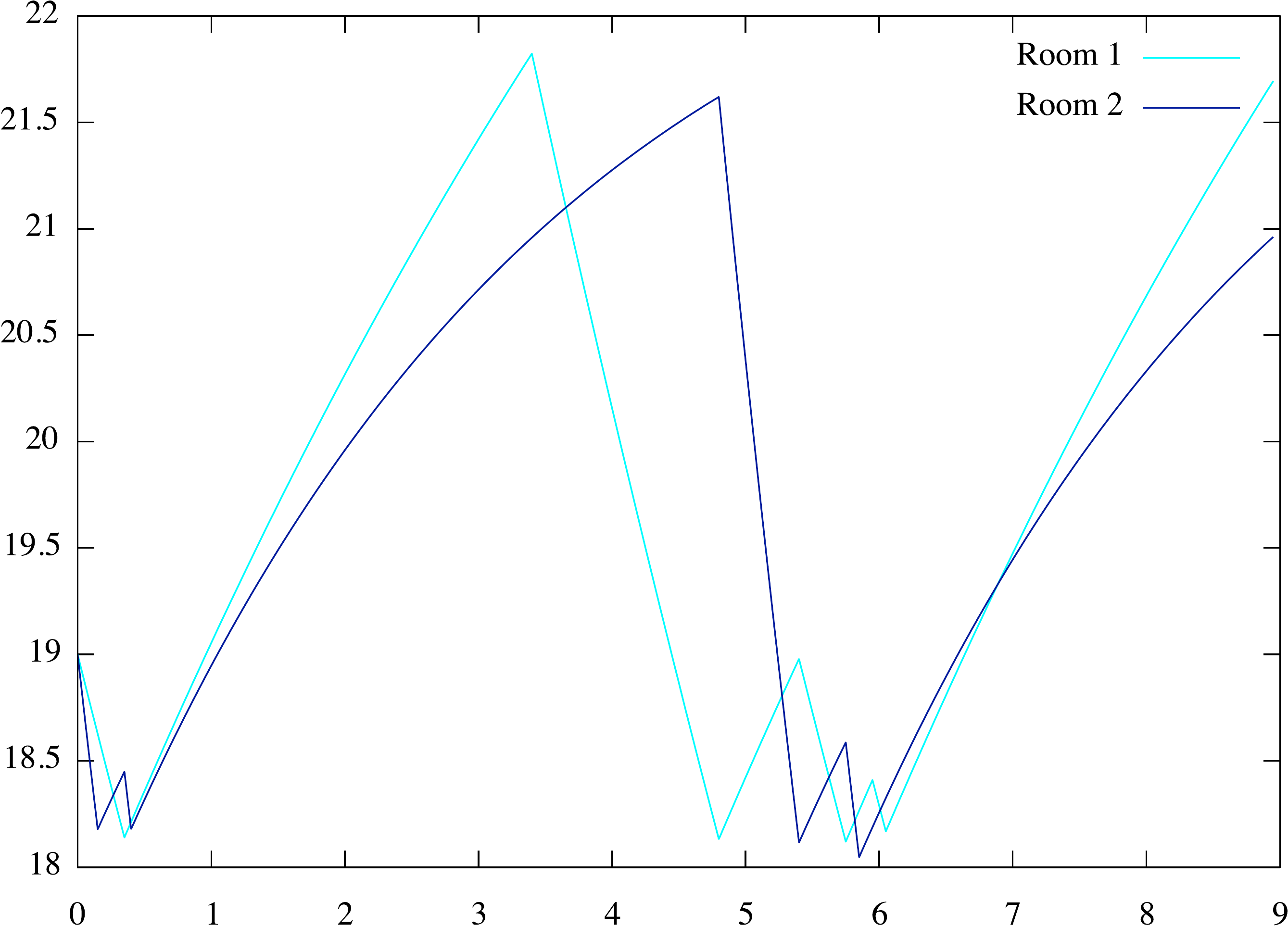}
    \caption{Comparison of temperature evolution under optimal and lazy control in an organisation consisting of two zones.
    The safe
    temperature is between 18$^\circ$C and 22$^\circ$C. On the left, a periodic
    controller with the minimum peak cost which was then optimised for
    the minimal average-cost. On the right, the behaviour of the lazy
    controller.
    The y-axis is temperature in $^\circ$C and the x-axis measures time in
    hours. The optimal controller used 3 modes and its minimum dwell time was 43 seconds.
    On the other hand, the lazy controller used 5 different modes and its minimum dwell time was 180 seconds.}
    \label{fig:example-run}
\end{figure}

We have implemented Algorithms \ref{alg:safe-schedule} and
\ref{alg:opt-schedule} using a basic implementation of the simplex algorithm
as their underlying linear program solver in Java. The tests were run on Intel
Core i5 1.7 GHz with 1GB memory available. The examples are based on the
model of an organisation with decoupled zones as in \cite{NBMJ11} and were
randomly generated with exactly the same parameters as described there.
We implemented also a simple lazy controller to compare its peak and
average energy consumption to our optimal one.
Simply asking the lazy controller to let the temperature oscillate around the minimum comfort 
temperature in each room  
is risky and causes high peak costs, so our ``lazy'' controller uses a different approach.
It switches any heater to its
minimum setting if its zone has reached a temperature in the top 5\% of its
allowable value range.
On the other hand, if the temperature in a zone is in the bottom 5\% of
its allowable value range, then the lazy controller finds and switches its
heater to the minimum setting that will prevent the temperature in that zone dropping
any further. However, before it does that, it first
checks whether there are any zones with their temperature above 10\% of their
allowable value range and switches them off first. This tries to minimise the number of heaters
being switched on at the same time and thus also tries to minimise the peak cost.

We have tested our systems for an organisation with eight zones and each heater
having six possible settings, which potentially gives $6^8 > 10^6$
possible modes. Zones parameters and their settings were generated using the
same distribution as described in \cite{NBMJ11} and the outside temperature was
set to 10$^\circ$C.
The simulation of the optimal
and the lazy controller was performed with a time step of three minutes and the duration of nine hours. 

First, in Figure \ref{fig:example-run} we can compare the difference in the behaviour
of the optimal controller as compared to the lazy one in the case of just two
zones.
In the case of the optimal controller, we can see that the temperature in each zone
stabilises around the lower safe bound by using a constant switching between various modes. 
On the other hand, for the lazy controller the
temperature oscillates between the lower and upper safe value, which wastes energy.
The peak cost was 15 kW for the optimal controller and 18.43 kW
for the lazy one, while the average energy usage was 13.4 kW and 15.7 kW,
respectively. This gives 23\% savings in the peak energy consumption and 17\% savings in the average energy consumption.
Note that any safe controller cannot use more than 16.9 kW of energy on the average, because otherwise it would
exceed the upper comfort temperature for one of the rooms, so the maximum possible savings in the average energy consumption 
cannot exceed 26\%. For a building with eight rooms, the running time of our algorithm
was between less than a second to up to a minute with an average 40 seconds,
depending on how many modes were necessary to ensure safe controllability of the system.
The lazy controller was found to have on the average 40\% higher peak cost than the optimal controller and 15\% higher average-cost.
In the extreme cases it had 70\% higher peak cost and 22\% higher average-cost. 
Again, the reason why the lazy controller did better in the average energy consumption than the peak consumption is
that the comfort zone is so narrow and any safe controller cannot waste too much energy
without violating the upper comfort temperature in one of the rooms.

 \section{Conclusions}
 \label{sec:conclusion}
 
We have proposed and analysed a subclass of hybrid automata %
with dynamics govern by linear differential equations and no guards on transitions. 
This model strictly generalises the models studied
by Nghiem et al. in \cite{NBMJ11} in the context of peak minimisation for
energy consumption in buildings. 
We gave a sufficient and necessary condition 
for the existence of a controller that keeps the state of the system within 
a given safe set at all times 
as well as
an algorithm that find such a controller in polynomial time.
We also analysed an extension of this model with costs per time unit associated with modes
and gave an algorithm 
that constructs a safe controller which minimises the peak cost, the average cost or any cost expressed
as a weighted sum of these two. 
Finally, we implemented some of these algorithms and showed how they perform in practice. 

From the practical point of view, the future work will involve turning the prototype implementation of the algorithms in this paper
into a tool. Our model can be extended by adding disturbances and interactions between zones to the dynamics of the model like in \cite{nghiem2012green}. This, however, would further complicate the already complicated formula given for the switching frequency of each mode of the
safe controller as defined in Algorithm \ref{alg:safe-schedule}. 
The special cases that could be looked at are the initial state being on the boundary of the safe set and
checking whether Theorem \ref{thm:safe-to-freq} also holds 
for all non-Zeno controllers not just for controllers with a positive minimum dwell time. 
An interesting problem left open 
is the decidability of finding a safe controller with the minimum dwell time
above a fixed constant.

\subsection*{Acknowledgments}
We would like to thank Rajeev Alur, Ashutosh Trivedi and Sasha Rubin for 
the discussions related to some aspects of this work. 
This research was supported by EPSRC grant {EP/G$050112$/$2$}.

\bibliographystyle{plain}
\bibliography{papers}

\newpage
\noindent{\Large \bf Appendix} \\
 \appendix
 
\section{Proof of Lemma \ref{lem:half-time}}

We first prove the following proposition.

\begin{proposition}
\label{prop:estimate}
If $\Hh$ is in the same mode $m$ during the time interval
$[t_0, t_0 + t]$ we have that  
$\bx_i(t_0) + t (b^m_i - a^m_i
\bx_i(t_0+t)) \leq \sgn(\bxdot_i(t_0)) \bx_i(t_0 + t) \leq \bx_i(t_0) +
t (b^m_i - a^m_i \bx(t_0))$ holds, where $\sgn$ is the signum function.  
\end{proposition}
\begin{proof}
Recall that $\sgn(x) = 1$ if $x > 0$, $\sgn(x)=-1$ if $x < 0$, and $\sgn(x)=0$ if $x=0$. 
Notice that $\bxdot_i(t) = b_i - a_i \bx_i(t)$ attains its minimum and maximum at the ends of the time interval $[t_0,t_0+t]$, because $\bx_i$ is monotone in $t$.
Therefore, if $\bxdot_i(t_0) > 0$, i.e. $\sgn(\bxdot_i(t_0)) = 1$ and $\bx_i$ is increasing in $t$,
then for all $c \in [t_0, t_0+t]$ we have $\bxdot_i(c) \leq \bxdot_i(t_0)$ and $\bxdot_i(c) \geq \bxdot_i(t_0+t)$.
From the mean value theorem, we know that for some $c \in [t_0, t_0 +t ]$ we have
$\bx_i(t_0+t) = \bx_i(t_0) + t \bxdot_i(c)$; the inequality follows and we proceed similarly in the other cases. 
\end{proof}

\noindent{\bf Lemma \ref{lem:half-time}}.
{\em 
For at least half of the time duration of every timed action of $\sigma$ which uses mode~$\in M'$,  $\bx_i(t) \geq \low_i + \gamma$
 holds.
}
\begin{proof}
Recall that
$\amax := \max_{m\in M'} {|a^m_i|}$,
  $\cmin := \min_{m\in M'}{|b^m_i - a^m_i \low_i|}$ and $\tmin :=
  \tmin(\sigma)$.
  Of course $\cmin > 0$, because ${b^m_i}/{a^m_i} \neq \low_i$ for $m \in
  M'$ and $\tmin > 0$, because $\sigma$ is \feasible. Let
  $\gamma := \frac{\tmin\cmin}{2 + \tmin\amax}$, which is $>0$.
  
Let us consider the $k$-th timed action of $\sigma$ such that
$m_k \in M'$. Of course we have $t_k \geq \tmin$.
Notice that $\gamma < \frac{1}{2}\tmin\cmin$, because $\tmin\amax > 0$ and  
also, easy calculations show $\gamma < \frac{\cmin}{\amax}$.
If $\bx_i(T_{k-1}) \geq \low_i + \gamma$ and $\bx_i(T_{k}) \geq \low_i +
\gamma$ then we are done, because $\bx_i$ is monotonic in the time interval $[T_{k-1},
T_k]$ and so $\bx_i(t) \geq \low_i +
\gamma$ would hold for the whole $k$-th timed action.
We will estimate the longest amount of time the system can be in the same
mode while $\bx_i(t) \in [\low_i,\low_i+\gamma]$ holds. 
First, notice that for every $m \in M'$ we
either have $b^m_i - a^m_i x > 0$ for all $x \in [\low_i, \low_i + \gamma]$ or
it is $b^m_i - a^m_i x < 0$ for all $x \in [\low_i, \low_i + \gamma]$.
Otherwise, there would be $x \in [\low_i, \low_i + \gamma]$ such that $b^m_i
- a^m_i x = 0$ and so $b^m_i - a^m_i \low_i = a^m_i (x - \low_i) \leq \amax \gamma<\amax
\frac{\cmin}{\amax} = \cmin$; a contradiction with the definition of $\cmin$.
Because we just showed that the process cannot coverage to any point in the
interval $[\low_i, \low_i + \gamma]$, either it reaches
the lower or upper boundary of this interval
or it runs out of the allocated amount of time $t_k$.

Now, assume that $b^m_i - a^m_i x < 0$ holds in that interval, i.e. the value of
$\bx_i(t)$ is decreasing in $t$. The amount of time $\bx_i(t) \in [\low_i,\low_i+\gamma]$
holds is the greatest if the value of variable
$\bx_i$ starts at $\low_i + \gamma$ and ends at $\low_i$.
Using Proposition \ref{prop:estimate}, we can estimate this time to be at most
$\gamma/(a^m_i l - b^m_i) < \frac{1}{2}\tmin\cmin / \cmin =
\frac{1}{2}\tmin$, which means that during the remaining time equal to
$t^m_k - \frac{1}{2}\tmin > \frac{1}{2}\tmin$, the value of $\bx_i(t)$
stays above $\low_i + \gamma$.

Finally, if $b^m_i - a^m_i x > 0$ holds in that interval, then
we can again estimate using Proposition \ref{prop:estimate} the time $\bx_i(t) \in [\low_i,\low_i+\gamma]$ 
can hold to be at most 
$\gamma/(b^m_i-a^m_i(\low_i+\gamma)) < \gamma/(\cmin-a^m_i\gamma) =
1/(\frac{\cmin}{\gamma} - a^m_i) = 1/(\frac{2+\tmin\amax}{\tmin} - a^m_i) <
\frac{1}{2}\tmin$.
Therefore, again the amount of time the value of $\bx_i$
stays outside of the interval
$[\low_i,\low_i+\gamma]$ is greater than the amount of time spent inside of it. 
\qed
\end{proof}

\section{Proof of Lemma \ref{lem:critical}}
\noindent{\bf Lemma \ref{lem:critical}}.
{\em Either there is a variable which is critical for all \safe frequency vectors or
there is \asafe frequency vector in which no variable is critical.}
\begin{proof}
Recall that
$\FF_i(\vf, y) := \sum_{m
\in M} \vf(m) (b^m_i - a^m_i y)$ and
for a frequency vector $\vf$, variable $\bx_i$ is called critical if $\FF_i(\vf, \low_i) = 0$
or $\FF_i(\vf, \upp_i) = 0$ holds.
Finally, a frequency vector $\vf$ is {\em \safe} if
for every variable $\bx_i$ the following conditions hold (I) $\FF_i(\vf, \low)
\geq 0$, and (II) $\FF_i(\vf, \upp) \leq 0$.

Now, let us assume that there is no variable which is critical for all \safe
frequency vectors. 
If so, for each variable $\bx_i$ we can find \asafe frequency vector $\vf_i$
for which $\bx_i$ is not critical. But if we consider the frequency vector $\vf =
\frac{1}{|N|}\sum_{j} \vf_j$, then no variable can be critical in $\vf$,
because $\FF_i(\frac{1}{N}\sum_{j} \vf_j, \low_i) =
\frac{1}{N}\sum_j \FF_i(\vf_j, \low_i) \geq \frac{N-1}{N} \low_i +
\frac{1}{N} \FF_i(\vf_i, \low_i) > \low_i$ and also
$\FF_i(\frac{1}{N}\sum_{j} \vf_j, \upp_i) =
\frac{1}{N}\sum_j \FF_i(\vf_j, \upp_i) \leq \frac{N-1}{N} \upp_i +
\frac{1}{N} \FF_i(\vf_i, \upp_i) < \upp_i$, which also proves that
such defined frequency vector $\vf$ would be \safe.
\qed 
\end{proof}

\newcommand{\tone}{{\bf 1}\xspace}
\newcommand{\tmone}{{\bf -1}\xspace}
\newcommand{\tzero}{{\bf 0}\xspace}
\newcommand{\tsafe}{{\bf S}\xspace}

\section{Proof of Theorem \ref{thm:min-dwell-time-hardness}}
\noindent{\bf Theorem \ref{thm:min-dwell-time-hardness}}.
{\em For a given MMS $\Hh$, hyperrectangular safe set $S$ described by two points
$\low, \upp$, starting point $\px_0 \in S$, checking whether
there exits a $S$-safe controller with minimum dwell time $\geq 1$ is \pspace-hard.
}
\begin{proof}
As mentioned before, the proof is similar to the \pspace-hardness proof in \cite{ATW12b} of
the discrete-time reachability in constant-rate MMS, which reduces from 
the acceptance problem for linear bounded automata (LBAs), so we first formally define LBAs.

An LBA $\Aa$ is a tuple $(\Sigma, Q,
q_0, q_A, \delta)$, where $\Sigma$ is a finite alphabet, $Q$ is a finite set of states, $q_0 \in Q$ is the initial
state, and $q_A \in Q$ is a distinguished accepting state, and
$\delta \subseteq Q \times \Sigma \times Q \times \Sigma \times \set{-1, 0,
  +1}$ is the transition relation.
  We can assume the alphabet $\Sigma$ to be the binary alphabet $\{0,1\}$. 
Let us explain the interpretation of the elements of the transition relation.
Let $\tau = (q, a, q', b, D) \in \delta$ be a transition.
If LBA $\Aa$ is in state $q \in Q$ and its (read/write) tape head reads
character $a$, then it writes character $b$ at the current cell and moves its
head in the direction $D$ (left if $D = -1$, right if $D = +1$, and unchanged
in $D=0$), and it changes the state to $q'$.
Let $w \in \Sigma^L$ be an input word.
Without loss of generality we assume that the LBA uses exactly $L$ tape cells,
which hold the whole input word of size $L$ at the very beginning.
Hence configuration of the LBA can be written as $(q, p, b_0 b_1 \ldots
b_{L-1})$ where
$q$ is the current state, $p$ is the position of head such that
$0 \leq p < L$, and $b_0 b_1 \ldots b_{L-1}$ is the current contents of the tape.
Notice that such an LBA has only $|Q|\cdot L\cdot|\Sigma|^L$ different
configurations and so if LBA does not enter the accepting state $q_A$ for a given input word 
after that many
steps then it never will. 

We show a reduction from the acceptance problem for LBAs to the problem of the
existence of a safe controller with its minimum dwell time $\geq 1$ for linear-rate MMSs.
For a given LBA $\Aa$ and input word $w = b_0 b_1 \ldots b_{L-1}$, we define LBA $\Hh_\Aa = (M, N, A, B)$
where there is one variable $\var{q, p}$ for each state $q \in Q$ and
head position $0 \leq p < L$, one variable $\var{i}$ for each input cell $b_i$ where $0 \leq i < L$,
and one
variable $\var{p, \tau}$ for each head position ${0 \leq p < L}$ and direction
$\tau \in \set{-1, 0, +1}$.
The safety condition $S$ simply requires that the value of all these variables at
all time belong to the interval $[-1, 1]$. Recall that 
$T_k(\sigma) \rmdef \sum_{i=1}^k t_i$ is the total time
elapsed up to step $k$ of the \controller $\sigma$ of $\Hh_\Aa$. A configuration $(q', p', b_0 b_1
\ldots b_{L-1})$ of machine $\Aa$ at step $k$ is encoded in the variables of
$\Hh_\Aa$ in a way that $\var{q, p}(T_k(\sigma)) > 0$ iff $q = q'$ \&  $p = p'$
and we have $\var{q, p}(T_k(\sigma)) < 0$ otherwise; and also for all $0 \leq i < L$ we have
$\var{i}(T_k(\sigma)) > 0$ iff the input cell $b_i = 1$ and we have $\var{i}(T_k(\sigma)) < 0$ iff $b_i  =
0$.
There is also a special variable $\bx_A$ which deals with the case when the input word is accepted
and special variable $\bx_T$ which as the only variable has a different safety interval $[-1,-0.9]$.  

We will now construct a gadget used in our reduction.
We say that variable $\bx_{v}$ (where $v \in \{i: 0\leq i<L\}\cup\{(q,p): q\in Q, 0\leq p < L\}$) in mode $m$ is of
(i) type \tone if $b^m_v = 2$ and $a^m_v = 1$, 
(ii) type \tmone if $b^m_v = -2$ and $a^m_v = 1$,
(iii) type \tzero if $b^m_v = 0$ and $a^m_v = 1/(11\cdot |Q|\cdot L\cdot
|\Sigma|^L)$.
Assume that the current mode is $m$, time is $t_0$ and the value of variable $\bx_v$ is safe, i.e. $x := \var{v}(t_0) \in [-1,1]$. Notice that
if the type of this variable is \tone in mode $m$, then after time $t$ its value
becomes $2 + (x - 2) e^{-t}$ which belongs to the safe set if $t=1$ and $x \leq 2-e
\approx -0.718$, but for $t \geq 1.1 > \ln 3$ its value is never safe.
Similarly for type \tmone, the new value is safe for $t=1$ and 
$x \geq e-2 \approx 0.718$, but after time $t \geq 1.1 > \ln 3$ it is never safe.
Finally, for a variable $\bx_v$ of type \tzero, we can compute that the relative change in the value of this
variable after time $t \leq 1.1 \cdot |Q|\cdot L\cdot
|\Sigma|^L$ to be $|e^{-a^m_vt} x - x|/|x| = 1 - e^{-a^m_v t} \leq 1 - e^{-\frac{1}{10}}
< 0.1$, i.e. its value does not change by more than $10\%$. Moreover, a
constant switching between a mode of type \tone and \tmone 
for some variable while spending in each mode amount of time $t = 1$ results in a trajectory that converges
to $\frac{4e^{-1}-2e^{-2}-2}{1-e^{-2}} \approx -0.924$ on odds 
steps and 
$\approx 0.924$ on even steps
independently of the starting point. 
Therefore, assuming the initial value of a variable is either 1 or -1, only modes of type \tone or \tmone are used, and the system is safe at all time,
the closest this variable can get to value $0$ is after the first step of length $t=1$. 
That value is $2-3e^{-1} \approx 0.896$ for a variable that starts at $-1$ and $-0.896$ for a variable that starts at $1$.
If we now allow that variable to switch to type \tzero as
well, then the closest such a process can get to $0$ 
is to let the just computed value decay 
towards $0$ by using modes where it has type \tzero only.
Its absolute value after time $t \leq 1.1 \cdot |Q|\cdot L\cdot
|\Sigma|^L$ would be still $> 0.896*0.9 \approx 0.8$ which is $ > 0.718$.
Assuming the number of timed actions in controller $\sigma$ does not exceed
$|Q|\cdot L\cdot|\Sigma|^L$, the minimum dwell time of each action is $\geq 1$ and
each mode has at least one variable of a nonzero type then we have the following.
A variable can remain safe in two consecutive timed actions 
if and only if its type changes from \tone to \tzero or \tmone,
from \tmone to \tzero or \tone, and from \tzero to 
\tzero or to type -$\mathbf{d}$ where $\mathbf{d}$ was the last nonzero type this variable had before \tzero.
If we interpret the value of a variable above $0.718$ as $1$ and below
$-0.718$ as $0$, then we can look at each timed action in a mode of type \tone, \tmone,
and \tzero as adding 1, subtracting 1, or keeping the value of that binary value constant, respectively.

Now, each transition  $\tau = (q, a, q', b, D) \in \delta$ and head position $p$
is simulated using two modes $M_{p, \tau}$ and $M'_{p, \tau}$.
Mode $M_{p, \tau}$ checks whether the letter in the $p$-th cell is $a \in \Sigma = \{0,1\}$,
while the mode $M'_{p, \tau}$ changes the content of the $p$-th cell 
to $b \in \{0,1\}$, and moves the head to a new position.
The rates of various variables in these modes are set in such a manner that
a schedule is safe if and only if it respects the transition structure of
$\Aa$.
The main features of the construction are the following.
\begin{itemize}
\item
  In mode $M_{p, \tau}$ the type of variable $\var{p}$ is \tmone if
  $a = 1$ and is \tone otherwise; the type of all other variables 
  is \tzero.
  This mode checks whether the character at head position $p$ is $a$.
\item
  In mode $M'_{p, \tau}$ the type of variable $\var{p}$ is \tzero if
  $a \not = b$.
  If $a = b$ then the type of variable $\var{p}$ is \tone.
  If $D = -1$ ($D = +1$) then variable $\var{q, p}$ has type \tmone
  and $\var{q', p-1}$ has type \tone  ($\var{q', p+1}$ has type \tone).
  While if $D = 0$ then $\var{q, p}$ has type \tzero for all $q \in Q$ and $0
  \leq p < n$.
  The type of all other variables is \tzero.
\item
  To make sure that mode $M_{p, \tau}$ is immediately followed by
  mode $M'_{p, \tau}$ in every safe run, the type of the variable
  $\var{p, \tau}$ is \tone in mode $M_{p, \tau}$, and \tmone in mode
  $M'_{p, \tau}$, while it is of type \tzero in every other mode.
\item
  The special variable $\bx_T$ has type \tzero in every mode and safe set $[-1,-0.9]$, so
  the system is safe as long as the decay from its initial value $-1$ 
  is not greater than 10\%, which does not happen before $1.1 \cdot |Q|\cdot L\cdot
|\Sigma|^L$ amount of time has elapsed. This guarantees that MMS $\Hh_\Aa$ becomes unsafe
  once we cannot guarantee that it follows the transitions of LBA $\Aa$ exactly.
\item 
  For each head position $0 \leq p < L$ we have two special modes   
  $M_{A,p}$ and $M'_{A,p}$, which deals with the case when LBA $\Aa$ enters 
  the accepting state $q_A$.
  In $M_{A,p}$ variable $\bx_{q_A,p}$ has type \tmone,
  variable $\bx_A$ has type \tone, and all other variables $\bx_v$ have a special safe type \tsafe such that
  $b^{M_{A,p}}_v = -1$ and $a^{M_{A,p}}_v = 1$.
  On the other hand, in  
  $M'_{A,p}$ variable 
  $\bx_{q_A,p}$ has type \tone,
  $\bx_A$ has type \tmone,  
  and all other variables have type \tsafe.
  Notice that once the system enters mode
  $M_{A,p}$ it can keep switching between modes 
  $M_{A,p}$ and $M'_{A,p}$ forever while being safe. 
  This is because, as it was pointed out before when all timed actions have delay $t=1$,
  the values of variables $\bx_{q_A,p}$ and $\bx_A$ in the limit keep switching between 
  $\approx -0.924$ and $\approx 0.924$ and all other variables, 
  which have type \tsafe, will converge from above to $-1$;  which belongs to the safe set of all of them. 
\end{itemize}
Note that each of the constructed modes has at least one variable of a nonzero type.
Let the initial state $\px_0$ of
$\Hh_\Aa$
be such that $\var{i}(0) = 1$ if the $i$-th input character $b_i = 1$,
$\var{q_0,0}(0) = 1$ (i.e. the initial state of $\Aa$ is $(q_0,0)$),  
and for all other variables we have $\var{v}(0) = -1$.
Notice that if the LBA $\Aa$ accepts the input word then there exists a $S$-safe controller
in MMS $\Hh_\Aa$ from the initial state $\px_0$, which at some point enters 
mode $M_{A,p}$ for some head position $p$ and keep switching between $M_{A,p}$
and $M'_{A,p}$. 
This has to happen before the value of variable $\bx_T$ becomes too close to $0$ to violate its
safety condition; until that moment $\Hh_\Aa$ models
precisely the configurations of LBA $\Aa$ and its transitions.
On the other hand, there is a safe feasible controller only if
$M_{A,p}$ is entered at some point, because otherwise  
variable $\bx_T$ will violate its safety condition eventually. 
So if a safe feasible controller for $\Hh_\Aa$ with minimum dwell
time $\geq 1$ exists, 
then LBA $\Aa$ has to enter the accepting state $q_A$ within
its first $|Q|\cdot L\cdot|\Sigma|^L$ timed actions. This shows that $\Aa$ accepts the input iff
$\Hh_\Aa$ has a safe feasible controller with minimum dwell time $\geq 1$. 
\qed

\end{proof}

\section{Proof of Corollary \ref{cor:cor}}

\noindent{\bf Corollary \ref{cor:cor}}.
{\em 
  Deciding whether there exists \afeasible $S$-safe \controller for a given MMS $\Hh$
  with a mode order specification graph $G$, initial mode $m_0$, a hyperrectangular safe set
  $S$ given by two points $\low$ and $\upp$ and an initial point $\low < \px_0 < \upp$ can be done in 
  polynomial time.  
}
\begin{proof}
Recall that a controller 
$\seq{(m_1, t_1), (m_2, t_2), \ldots}$
respects the mode order specification graph $G$ with initial mode $m_0$, if for all $i \geq 1$ we have $(m_i,m_{i+1}) \in G$ and $m_1 = m_0$.
Notice that the system $\Hh$ under any \feasible $S$-safe \controller will eventually end up in one of the strongly connected components (SCC) of the graph $G$ reachable from $m_0$, because the controller is non-Zeno and each timed action takes only a finite amount of time.
Also note that the safe controller returned by
Algorithm \ref{alg:safe-schedule}
returns a periodic controller which cycles over
all the modes given to it in exactly the same order as 
they were passed. 
Therefore, we can make sure the controller returned satisfies the mode order specification $G$
by passing the mode sequence in a particular order.
For an SCC $C$ of $G$ consisting of modes $C = \seq{m'_1, m'_2, \ldots, m'_k}$ that sequence of modes, denoted by $\rho_C$, is as follows: it starts 
at $m'_1$, then follows any path of modes in $G$ to $m'_2$, $\ldots$, then any path of modes to $m'_k$, 
and finally any path of modes to $m'_1$; all these paths exist because $C$ is an SCC. 
The sequence of modes $\rho_C$ can repeat some modes, but it satisfies 
the mode order specification graph $G$, each mode of $C$ occurs at least once and no mode outside $C$ occurs along $\rho_C$.
It is quite easy to see that there is 
\afeasible $S$-safe \controller for an initial state in the interior of the safe set for the set of modes $C$ iff
there is one for the sequence of modes $\rho_C$.

Now, for each SCC $C$ of $G$ reachable from the initial mode $m_0$ we 
check using Algorithm \ref{alg:safe-schedule} whether there is  
\afeasible $S$-safe \controller for the mode sequence $\rho_C$ and initial point $\px_0$.
If there is no such SCC then there is no \feasible $S$-safe \controller
which respects the mode order specification from $\px_0$ either, because
while using such a controller the system $\Hh$ has to eventually repeat only modes from a single SCC of $G$.
On the other hand, if there is such an SCC $C$, then we construct
\afeasible $S$-safe \controller from $\px_0$ as follows. 
First, we find any path in $G$ from $m_0$ to the very first mode in the
mode sequence $\rho_C$. 
We create a finite timed actions sequence based on this path where the time delay of each mode is set to such a small value
that when starting at $\px_0$ the system will still remain within the safe set $S$ at the very end of it. 
Such a value always exists when the initial point of $\Hh$ is in the interior of the safe set. 
To be precise, it suffices to set it to
$\min_{i\in I} \left(
\frac{\min\{\px_0(i) - \low_i, \upp_i - \px_0(i)\}}{\max_m |b^m_i -
a^m_i \px_0(i)|} \right)$.
Let the point reached at the end of this finite timed sequence be $\px_1 \in S$ and, because the coordinates of that point are likely to be irrational, 
let $\px_l$ and $\px_u$ be any two points with rational coordinates such that $\low < \px_l \leq \px_1 \leq \px_u < \upp$ holds.
In other words, $\px_l$ and $\px_u$ are simply some polynomial size lower and upper bounds on the coordinates of the point $\px_1$.  
Notice that the \feasible controller for the mode sequence $\rho_C$ that
we found earlier may not be safe when the system starts at $\px_1$ instead of $\px_0$,
because the value of $s$ may need to be smaller for the system to remain safe.
The new value of 
$s$ 
should be the minimum of the value of $s$ for the mode sequence $\rho_C$ when the initial point is $\px_l$ and when it is $\px_u$.
Finally, once we combine the finite timed action sequence, which starts at the mode $m_0$ and $\px_0$, with the
\feasible \controller for the mode sequence $\rho_C$, which is safe for any initial point between $\px_l$ and $\px_u$, we will
get \afeasible $S$-safe \controller that respects the mode order specification graph $G$.
\qed 
\end{proof}

Notice that if one would like to extend the model and allow the system to remain in the same mode forever, 
instead of forcing it to constantly switch between modes, it suffices to
add in $G$ an edge from each mode to itself.

\newpage
\section{Simpler Algorithm for Finding an Optimal Controller}
\label{app:alg} 
\begin{algorithm}
\label{alg:opt-schedule-simpler}
\caption{Finds an optimal $S$-safe \feasible controller from a given $\px_0 \in
S$.} \KwIn{A priced MMS $\Hh$, two points $\low$ and $\upp$ that define a 
hyperrectangle $S = \{\px : \low \leq \px \leq \upp\}$ and an initial point
$\px_0 \in S$ such that $\low < \px_0 < \upp$, and constants $\cavg$ and
$\cpeak$ which define the cost of a controller.} \KwOut{NO if no $S$-safe
\feasible controller exists from $\px_0$, and an periodic such controller $\sigma$
for which $\cpeak \peakcost(\sigma) + \cavg \avgcost(\sigma)$ is
minimal, otherwise.
} 

Modify Algorithm
\ref{alg:safe-schedule} by adding
the objective function
{\em Minimise $\sum_{m\in M'}{\vf_m \pi(m)}$}
to the linear program at line 
\ref{alg-line:safe-LP2}. Let $\optavgcost(M')$ 
be the value of this objective when Algorithm \ref{alg:safe-schedule} is
called for the set of modes $M'$. 

Let $P = \{ \pi(m) : m \in M \}$ be the set of all different costs of modes of $\Hh$. 
(Notice that only these costs can be potential peak costs.) For a given $p$ let $M_{\leq
p}$ denote the set of modes with cost at most $p$. 

Iterate over $p \in P$ and find the one with the smallest value of $\cpeak p + \cavg \optavgcost(M_{\leq p})$ and denote
it by $p^*$.

{\bf return} the periodic controller returned by
the modified version of Algorithm \ref{alg:safe-schedule} called for the set of modes $M_{\leq p^*}$.
\end{algorithm}

\end{document}